\documentclass[12pt]{article}

\usepackage{graphicx}
\usepackage{natbib}
\usepackage{amsmath,amsthm,amssymb,bm,fancybox,multirow,hhline,psfrag,url,comment}

\newtheorem{Lem}{Lemma}[section]
\newtheorem{Pro}{Proposition}[section]

\newtheorem{Example}{Example}[section]

\def\E{{E}}

\def\l{{l}}

\begin{document}

\def\spacingset#1{\renewcommand{\baselinestretch}%
{#1}\small\normalsize} \spacingset{1.5}


\begin{center}
\textbf{\Large Sparse estimation via non-concave penalized likelihood in factor analysis model}
\end{center}
\begin{center}

\large {Kei Hirose and Michio Yamamoto}
\end{center}
\begin{center}
{\it {\small
Division of Mathematical Science, Graduate School of Engineering Science, Osaka University,\\
1-3, Machikaneyama-cho, Toyonaka, Osaka, 560-8531, Japan 
}}
\end{center}

\begin{abstract}
We consider the problem of sparse estimation in a factor analysis model.  A traditional estimation procedure in use is the following two-step approach: the model is estimated by maximum likelihood method and then a rotation technique is utilized to find sparse factor loadings.  However, the maximum likelihood estimates cannot be obtained when the number of variables is much larger than the number of observations.  Furthermore, even if the maximum likelihood estimates are available, the rotation technique does not often produce a sufficiently sparse solution.  In order to handle these problems, this paper introduces a penalized likelihood procedure that imposes a nonconvex penalty on the factor loadings.   We show that the penalized likelihood procedure can be viewed as a generalization of the traditional two-step approach, and the proposed methodology can produce sparser solutions than the rotation technique.  A new algorithm via the EM algorithm along with coordinate descent is introduced to compute the entire solution path, which permits the application to a wide variety of convex and nonconvex penalties.  Monte Carlo simulations are conducted to investigate the performance of our modeling strategy.  A real data example is also given to illustrate our procedure.  
\end{abstract}

\noindent%

\spacingset{1.5}
\section{Introduction}

Factor analysis provides a useful tool  for exploring the covariance structure among a set of observable variables by construction of a smaller number of unobservable variables.  
In exploratory factor analysis, a traditional estimation procedure in use is the following two-step approach: the model is estimated by the maximum likelihood method 
with the use of efficient algorithms 
(e.g., \citealp{joreskog1967some}; \citealp{jennrich1969newton}; \citealp{clarke1970rapidly}), and then rotation techniques, such as the varimax method \citep{kaiser1958varimax} and the promax method \citep{hendrickson1964promax}, are utilized to find sparse factor loadings.   
However, it is well known that the maximum likelihood method often yields unstable estimates because of overparametrization (e.g., \citealp{akaike1987factor}).  In particular, the above-mentioned algorithms cannot often be used when the number of variables is much larger than the number of observations.  Furthermore, even if the maximum likelihood  estimates are available, the rotation technique does not often produce a sufficiently sparse solution.  In such cases, a penalized likelihood procedure that produces the sparse solutions, such as the lasso \citep{Tibshirani:1996}, can handle these issues.  


The lasso \citep{Tibshirani:1996}  has received much attention as a powerful tool for sparse estimation in a variety of statistical modeling and machine learning, including generalized linear models, Gaussian graphical models, and the support vector machine (e.g., \citealp{hastie2004entire,yuan2007model,Friedmanetal:2010}). 
It is, however, well known that the lasso is biased and estimates an overly dense model \citep{Zou:2006,zhao2007model,Zhang:2010}. Typically, a penalization methods via nonconvex penalties can achieve sparser models than the lasso.
  Recently, a number of researchers have presented various nonconvex methods:  bridge regression \citep{FrankFriedman:1993,Fu:1998}, smoothly clipped absolute deviation (SCAD, \citealp{FanLi:2001}), minimax concave penalty (MC+, \citealp{Zhang:2010}), and generalized elastic net \citep{Friedman2008}.  
  Several efficient algorithms to obtain the entire solutions have also been proposed   (e.g., local linear approximation, \citealp{zou2008one}; generalized path seeking, \citealp{Friedman2008}; penalized linear unbiased selection algorithm, \citealp{Zhang:2010}; coordinate descent algorithm, \citealp{breheny2011coordinate,Mazumderetal:2009}).

In the framework of factor analysis models, a few researchers have discussed the penalized likelihood procedure. \citet{Ningetal:2011} and \citet{Choietal:2011} applied the lasso-type penalization procedure to obtain sparse factor loadings, and numerically showed that the penalization method often outperformed the above-mentioned two-step traditional technique.  
However,  the relationship between the penalized likelihood procedure and the two-step traditional approach has not yet been discussed.
Moreover, although the ordinary lasso can  estimate an overly dense model in penalized likelihood factor analysis \citep{Choietal:2011,hirose2012variable} as discussed earlier,  
the penalized likelihood methods via the nonconvex penalties that produce sparser solutions than the lasso  have not yet been developed.

In the present paper, a new penalization method via nonconvex penalties in factor analysis models is proposed.  We show that penalized likelihood method can be viewed as a generalization of the traditional two-step approach, i.e., the maximum likelihood estimates with the rotation technique can be obtained by the penalized likelihood procedure under certain conditions. The proposed methodology can produce sparser solutions than the rotation technique with maximum likelihood estimates by changing the regularization parameter.  
 In order to produce entire solutions of factor loadings and unique variances,  a new pathwise algorithm via the EM algorithm \citep{rubin1982algorithms} along with coordinate descent for nonconvex penalties \citep{Mazumderetal:2009} is introduced, which permits the application to a wide variety of convex and nonconvex penalties including the lasso, SCAD, and MC+ family.   A  package {\tt fanc} in {\tt R} \citep{R:2010}, which implements our algorithm, is available from Comprehensive R Archive Network (CRAN) at {\tt http://cran. r-project.org/web/packages/fanc/index.html}.   

The remainder of this paper is organized as follows: Section 2 describes the traditional estimation procedure in factor analysis.   In Section 3, we introduce a penalized likelihood factor analysis via nonconvex penalties, and show that penalized likelihood procedure can be viewed as a generalization of the maximum likelihood method with the rotation technique.  Section 4 provides a new algorithm  based on the EM algorithm and coordinate descent to obtain the entire solution path.   In Section 5, we present numerical results for both artificial and real datasets. Some concluding remarks are given in Section 6.


\section{Traditional Estimation Procedure in Factor Analysis}
\subsection{Maximum Likelihood Factor Analysis}
Let $\bm{X}=(X_1,\dots,X_p)^T$ be a $p$-dimensional observable random vector with mean vector $\bm{\mu} $ and variance-covariance matrix $\bm{\Sigma}$. The factor analysis model (e.g., \citealp{mulaik2010foundations}) is
\begin{equation*}
\bm{X} =\bm{\mu} + \bm{\Lambda} \bm{F}+\bm{\varepsilon}  , \label{model1}
\end{equation*}
where $\bm{\Lambda} =(\lambda_{ij})$ is a $p \times m$  matrix of factor loadings, and $\bm{F} = (F_1,\cdots,F_m)^T$ and $\bm{\varepsilon}  = (\varepsilon_1,\cdots, \varepsilon_p)^T$ are unobservable random vectors. The elements of $\bm{F}$  and $\bm{\varepsilon}$  are called common factors and unique factors, respectively. It is assumed that the common factors $\bm{F}$ and the unique factors $\bm{\varepsilon}$ are multivariate-normally distributed with  $\E(\bm{F} )=\bm{0}$,  $\E(\bm{\varepsilon} )=\mathbf{0}$,  $\E(\bm{F}\bm{F}^T)=\mathbf{I}_m$, $\E(\bm{\varepsilon} \bm{\varepsilon} ^T)=\bm{\Psi} $, and are independent (i.e., $\E(\bm{F} \bm{\varepsilon} ^T)=\bm{O}$), where $\mathbf{I}_m$ is the identity matrix of order $m$, and $\bm{\Psi} $ is a $p \times p$ diagonal matrix with $i$-th diagonal element $\psi_i$, which is called unique variance.  Under these assumptions, the observable random vector $\bm{X}$ is multivariate-normally distributed with variance-covariance matrix $\bm{\Sigma} = \bm{\Lambda} \bm{\Lambda}^T+\bm{\Psi}$.

Suppose that we have a random sample of $N$ observations $\bm{x}_1,\cdots,\bm{x}_N$ from the $p$-dimensional normal population $N_p(\bm{\mu} , \bm{\Lambda} \bm{\Lambda}^T+\bm{\Psi})$. 
The log-likelihood function is given by
\begin{eqnarray}
\ell(\bm{\Lambda}} ,{\bm{\Psi})= -\frac{N}{2}  \left[ p\log(2\pi)+\log |\bm{\Lambda} \bm{\Lambda}^T+\bm{\Psi}| + \mathrm{tr}\{ (\bm{\Lambda} \bm{\Lambda}^T+\bm{\Psi})^{-1} \bm{S} \} \right] , \label{taisuuyuudo}
\end{eqnarray}
where 
$\mathbf{S}=(s_{ij})$ is the sample variance-covariance matrix.

The maximum likelihood estimates of ${\bm{\Lambda} }$ and ${\bm{\Psi} }$ are given as the solutions of $\partial \ell(\bm{\Lambda} ,\bm{\Psi} ) / \partial \bm{\Lambda}  = \mathbf{0}$ and $\partial \ell(\bm{\Lambda} ,\bm{\Psi} ) / \partial \bm{\Psi}  = \mathbf{0}$. Since the solutions cannot be expressed in a closed form,  several researchers have proposed iterative algorithms to obtain the maximum likelihood estimates $\hat{\bm{\Lambda} }_{\mathrm{ML}}$ and $\hat{\bm{\Psi} }_{\mathrm{ML}}$ (e.g., \citealp{joreskog1967some}; \citealp{jennrich1969newton}; \citealp{clarke1970rapidly}; \citealp{rubin1982algorithms}).     

\subsection{Rotation Techniques}
The factor loadings have a rotational indeterminacy, because both $\bm{\Lambda} $ and  $\bm{\Lambda} \mathbf{T}$ generate the same covariance matrix $\bm{\Sigma}$, where  $\mathbf{T}$ is an arbitrary orthogonal matrix.  
The rotation technique (e.g., varimax method,  \citealp{kaiser1958varimax}) has been widely used to find the matrix $\mathbf{T}$ which gives a meaningful relation between items and factors.   
Suppose that $Q(\bm{\Lambda})$ is an orthogonal rotation criterion at $\bm{\Lambda}$.     
The criterion is minimized over all orthogonal rotations with an initial loading matrix being $\hat{\bm{\Lambda}}_{\rm ML}$, i.e.,
 \begin{equation}
\min_{\bm{\Lambda}} Q(\bm{\Lambda}),  \mbox{ subject to } \bm{\Lambda}=\hat{\bm{\Lambda}}_{\rm ML} \mathbf{T} \mbox{ and } \mathbf{T}^T\mathbf{T}=\mathbf{I}_m. \label{problem_rotation_mle0}
\end{equation}

For ease of comprehension,  here and throughout this paper, the criterion $Q(\bm{\Lambda})$ is given by the component loss criterion $Q(\bm{\Lambda}) = \sum_{i=1}^p\sum_{j=1}^mP( |\lambda_{ij}|)$ \citep{jennrich2004rotation,jennrich2006rotation}. 
Note that it is not necessary to assume that $Q(\bm{\Lambda})$ is a component loss criterion for constructing our modeling procedure. The problem in (\ref{problem_rotation_mle0}) is then expressed as
 \begin{equation}
\min_{\bm{\Lambda}} \sum_{i=1}^p\sum_{j=1}^mP( |\lambda_{ij}|),  \mbox{ subject to } \bm{\Lambda}=\hat{\bm{\Lambda}}_{\rm ML} \mathbf{T} \mbox{ and } \mathbf{T}^T\mathbf{T}=\mathbf{I}_m. \label{problem_rotation_mle}
\end{equation}

\begin{Example}\label{example:pss}
Suppose that the maximum likelihood estimates of factor loadings are
\begin{equation*}
\hat{\bm{\Lambda}}_{\rm ML} =
\left(
\begin{array}{rrrrrr}
0.60&0.60&0.60&0.60&0.60&0.60\\
0.60&0.60&0.60&-0.60&-0.60&-0.60\\
\end{array}
\right)^T.
\label{pss00}
\end{equation*}
If $P(|\theta|)$ is a $L_1$ loss function, i.e., $P(|\theta|) = |\theta|,$ the factor loadings are estimated by
\begin{equation}
\left(
\begin{array}{rrrrrr}
0.00&0.00&0.00&0.82&0.82&0.82\\
0.82&0.82&0.82&0.00&0.00&0.00\\
\end{array}
\right)^T.
\label{pss}
\end{equation}
In this way, the $L_1$ loss function tends to produce a sparse solution.
\end{Example}

The  factor loading in (\ref{pss}) possesses a perfect simple structure, that is, each row has at most one nonzero element. It is shown that the $L_1$ loss function $P(|\theta|) = |\theta|$  can recover perfect simple structure whenever it exists (\citealp{jennrich2004rotation,jennrich2006rotation}).  Note that  even if the true factor loadings possess the perfect simple structure,  the maximum likelihood estimates cannot have a perfect simple structure because of sampling error. 

\begin{Example}\label{example:pss2}
Suppose that the data is generated from $\bm{x} \sim N_8(\bm{0},\bm{\Lambda} \bm{\Lambda}^T + \bm{\Psi} )$ with $N= 50$, where true factor loadings $\bm{\Lambda}$ and unique variances $\bm{\Psi}$ are given by (\ref{pss}) and $\bm{\Psi}=0.32 \hspace{0.5mm} \mathbf{I}_8$, respectively.  The parameters were estimated by the maximum likelihood method and then the rotation technique via the $L_1$ loss function $P(|\theta|) = |\theta|$ was applied.  The estimated factor loadings are given by
\begin{equation*}
\left(
\begin{array}{rrrrrr}
0.10&0.08&0.08&0.75&0.79&0.70\\
0.87&0.83&0.74&0.12&0.07&0.00\\
\end{array}
\right)^T.
\label{pssml}
\end{equation*}
It can be seen that only one out of six parameters was correctly estimated by exactly zero.  
\end{Example}
As shown in Example \ref{example:pss2}, the rotation technique based on the maximum likelihood estimates can often yield an overly dense model.  In addition, the maximum likelihood estimates cannot often be obtained when  the number of variables is much larger than the number of observations.  In order to enhance the sparsity and deal with high-dimensional data, we employ a penalized likelihood procedure.  

\section{Penalized Likelihood Factor Analysis via Nonconvex Penalties}\label{Comparison with rotation technique}

Assume that  the maximum likelihood estimates $\hat{\bm{\Lambda} }_{\mathrm{ML}}$ are unique if  the indeterminacy of the rotation in $\hat{\bm{\Lambda} }_{\mathrm{ML}}$ is taken out.  An example of the condition for identification of factor loadings is given by Theorem 5.1 of \citet{anderson1956statistical}. The problem in (\ref{problem_rotation_mle}) is then expressed as
 \begin{equation}
\min_{\bm{\Lambda} } \sum_{i=1}^p\sum_{j=1}^mP( |\lambda_{ij}|), \mbox{ subject to} \quad   \ell(\bm{\Lambda},\bm{\Psi}) = \hat{\ell}, \label{problem_rotation_mle2}
\end{equation}
where $\hat{\ell}=\ell(\hat{\bm{\Lambda} }_{\mathrm{ML}},\hat{\bm{\Psi} }_{\mathrm{ML}}) $.  

The sparsity may be enhanced by modifying the problem  in (\ref{problem_rotation_mle2}) as follows:
 \begin{equation}
\min_{\bm{\Lambda}}\sum_{i=1}^p\sum_{j=1}^mP( |\lambda_{ij}|), \mbox{ subject to} \quad   \ell(\bm{\Lambda},\bm{\Psi}) \ge \ell^*, \label{problem_rotation_pmle}
\end{equation}
where $\ell^*$  ($\ell^* \le \hat{\ell}$) is a constant value. The value  $\ell^*$ controls the balance between the fitness of data and sparseness.  When  $\ell^*=\hat{\ell}$, the solution  coincides with the maximum likelihood estimates.  On the other hand,  $\bm{\Lambda}=\bm{O}$ if $\ell^* \rightarrow -\infty$.  

The problem in  (\ref{problem_rotation_pmle}) can be solved by maximizing the following penalized log-likelihood function $\ell_{\rho}(\bm{\Lambda},\bm{\Psi}) $:
\begin{equation}
\ell_{\rho}(\bm{\Lambda},\bm{\Psi})  =  \ell(\bm{\Lambda},\bm{\Psi})  -N   \sum_{i=1}^p\sum_{j=1}^m \rho P(|\lambda_{ij}|)
, \label{problem_rotation_pmle3_adj}
\end{equation}
where $\rho > 0$  is a regularization parameter.  Here $P(\cdot)$ can be viewed as a penalty function.  The regularization parameter $\rho$  controls the amount of shrinkage; that is, the larger the value of $\rho$, the greater the amount of shrinkage.  

The  following Proposition suggests that penalized likelihood procedure can be viewed as a generalization of  the maximum likelihood method with the rotation technique.
\begin{Pro} 
When $\rho \rightarrow +0$, the solution in  (\ref{problem_rotation_pmle3_adj}) becomes the maximum likelihood estimates with the rotation technique in (\ref{problem_rotation_mle2}) if the solution path in (\ref{problem_rotation_pmle3_adj}) is continuous near the maximum likelihood estimates.  
\end{Pro}

  The $L_1$ penalization procedures, such as the lasso, can yield sparse solutions for some values of $\rho$.  However, the lasso is biased and estimates an overly dense model (e.g., \citealp{Zou:2006}; \citealp{Zhang:2010}).  In this paper, we apply the nonconvex penalties to achieve sparser models than the lasso.  Here are two popular nonconvex penalties.  
\begin{itemize}
\item The SCAD \citep{FanLi:2001} is a nonconvex penalty which is taken sparsity, continuity, and unbiasedness into consideration simultaneously:
\begin{eqnarray*}
P'(\theta;\rho;\gamma)=I(\theta \le \rho) + \frac{(\gamma\rho-\theta)_+}{(\gamma-1)\rho}I(\theta > \rho) \quad \mbox{ for $\gamma>2$} .
\end{eqnarray*}
\item The MC+ \citep{Zhang:2010}  is defined by   
\begin{eqnarray*}
\rho P(|\theta|;\rho;\gamma)&=&\rho \int_0^{|\theta|}\left(1-\frac{x}{\rho\gamma}\right)_+dx\\
&=&\rho \left(|\theta|-\frac{\theta^2}{2\rho\gamma}\right) I(|\theta| < \rho\gamma) + \frac{\rho^2\gamma}{2}I(|\theta| \ge \rho\gamma).
\end{eqnarray*}
For each value of $\rho>0$, $\gamma \rightarrow \infty$ yields soft threshold operator (i.e., lasso penalty) and $\gamma \rightarrow 1+$ produces hard threshold operator.
\end{itemize}
\begin{Example}\label{example:pss3}
We used the same dataset as in Example \ref{example:pss2}, and the entire solution of the lasso and MC+ with $\gamma=7.6$ was computed.  
The solution path of the lasso and  MC+ ($\gamma=7.6$) are depicted in Figure \ref{fig:example1}.  The solid line corresponds to the true nonzero elements, and the dashed line corresponds to the zero elements.      
 It can be seen that the lasso cannot recover the true model no matter what value of $\rho$ is chosen, whereas the MC+ was able to select the correct model when $\rho \in [0.20, 0.37]$.  
\end{Example}
\begin{figure}[t]
\begin{center}
    \includegraphics[width=140mm]{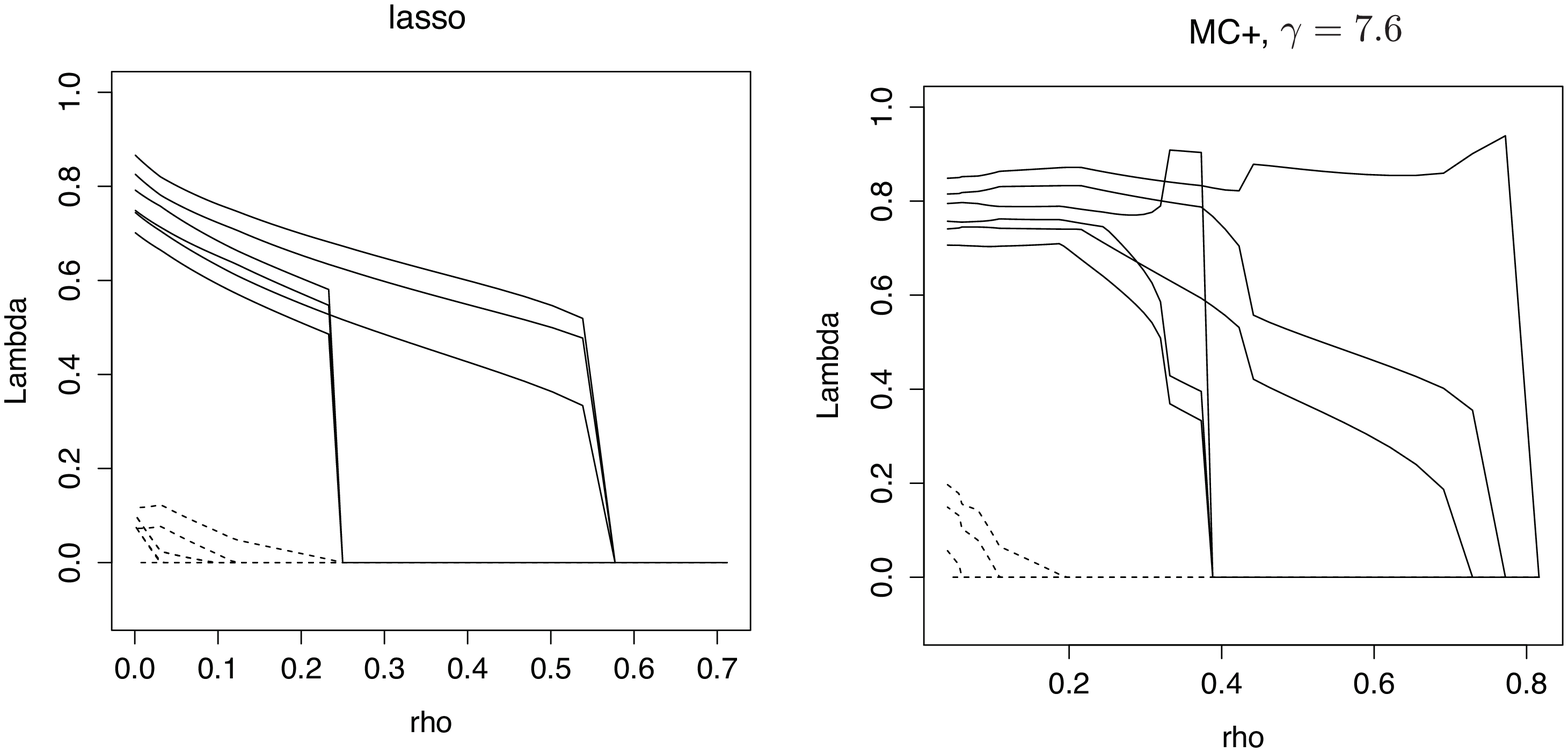}
\caption{Solution path of the lasso and  MC+ ($\gamma=7.6$). The solid line corresponds to the true nonzero elements, and the dashed line corresponds to the zero elements.    }
\label{fig:example1}
\end{center}
\end{figure}

\section{Algorithm}
It is well known that the solutions estimated by the lasso-type regularization methods are not usually expressed in a closed form because the penalty term includes a nondifferentiable function.    In regression analysis, a number of researchers have proposed fast algorithms to obtain the entire solutions (e.g., Least angle regression, \citealp{Efronetal:2004}; Coordinate descent algorithm, \citealp{Friedmanetal:2007}; Generalized path seeking, \citealp{Friedman2008}).    
The coordinate  descent algorithm is known as a very fast algorithm \citep{Friedmanetal:2010} and can also  be applied to a wide variety of convex and nonconvex penalties \citep{Mazumderetal:2009}, so that  we employ the coordinate descent algorithm to obtain the entire solutions.  

In the coordinate descent algorithm, each step is fast if an explicit formula for each coordinate-wise maximization is given,  whereas the log-likelihood function in (\ref{taisuuyuudo}) may not lead to the explicit formula.    
In order to derive the explicit formula, we apply the EM algorithm  \citep{rubin1982algorithms} to the penalized likelihood factor analysis, and the coordinate descent algorithm is utilized to maximize the nonconcave function in the maximization step of the EM algorithm.  Because the complete-data log-likelihood function takes the quadratic form, the explicit formula for each coordinate-wise maximization is available. 

\subsection{Update Equation for Fixed Regularization Parameter}
First, we give the update equations of factor loadings and unique variances when $\rho$ and $\gamma$ are fixed.  Suppose that $\bm{\Lambda}_{\rm old}$ and $\bm{\Psi}_{\rm old}$ are the current values of factor loadings and unique variances.   The parameters can be updated by maximizing the expectation of the complete-data penalized log-likelihood function with respect to $\bm{\Lambda}$  and $\bm{\Psi}$:
\begin{eqnarray}
E[\l_{\rho}^{C} (\bm{\Lambda},\bm{\Psi})] &=&- \frac{N}{2} \sum_{i=1}^p \log \psi_i - \frac{N}{2} \sum_{i=1}^p \frac{s_{ii} - 2\bm{\lambda}_i^T\bm{b}_i+ \bm{\lambda}_i^T \bm{A}\bm{\lambda}_i}{\psi_i}
 \cr
&&
   - \frac{N\rho}{2}  \sum_{i=1}^p\sum_{j=1}^mP( |\lambda_{ij}|)+{\rm const.,}  \label{ECL}
\end{eqnarray}
where $\bm{b}_i = \bm{M}^{-1}\bm{\Lambda}_{\rm old}^T\bm{\Psi}_{\rm old }^{-1}\bm{s}_i$ and $\bm{A}=\bm{M} ^{-1} + \bm{M}^{-1}\bm{\Lambda}_{\rm old}^T\bm{\Psi}_{\rm old }^{-1}\bm{S}\bm{\Psi}_{\rm old }^{-1}\bm{\Lambda}_{\rm old }\bm{M}^{-1}$. Here $\bm{M} = \bm{\Lambda}_{\rm old }^T\bm{\Psi}_{\rm old }^{-1}\bm{\Lambda}_{\rm old } + \mathbf{I}_m$, and $\bm{s}_i$ is the $i$-th column vector of $\bm{S}$.  The derivation of the complete-data penalized log-likelihood function is described in Appendix A. The new parameter $(\bm{\Lambda}_{\rm new}, \bm{\Psi}_{\rm new})$ can be computed by maximizing the complete-data penalized log-likelihood function, i.e., 
\begin{eqnarray}
(\bm{\Lambda}_{\rm new}, \bm{\Psi}_{\rm new}) = {\rm arg}\max_{\bm{\Lambda}, \bm{\Psi}} E[\l_{\rho}^{C} (\bm{\Lambda},\bm{\Psi})] .\label{maxECL}
\end{eqnarray}
The solutions in (\ref{maxECL}) are not usually expressed in a closed form because the penalty term includes a nondifferentiable function, so that the coordinate descent algorithm is utilized.


Let  $\tilde{\bm{\lambda}}_{i}^{(j)} $ be an ($m-1$)-dimensional vector  $( \tilde{\lambda}_{i1},\tilde{\lambda}_{i2},\dots,\tilde{\lambda}_{i(j-1)},\tilde{\lambda}_{i(j+1)},\dots,\tilde{\lambda}_{im})^T$.  The parameter $\lambda_{ij}$ can be updated by maximizing (\ref{ECL}) with the other parameters $\tilde{\bm{\lambda}}_{i}^{(j)} $ and $\bm{\Psi}$ being fixed, i.e.,  we solve the following problem:
\begin{eqnarray}
\tilde{\lambda}_{ij} &=& {\rm arg} \min_{\lambda_{ij}}  \frac{1}{2\psi_i} \left\{a_{jj}\lambda_{ij}^2  - 2\left(b_{ij} - \sum_{k \ne j} a_{kj} \tilde{\lambda}_{ik} \right)\lambda_{ij} \right\}  + \rho   P( |\lambda_{ij}|) \cr
&=& {\rm arg} \min_{\lambda_{ij}}\frac{1}{2}   \left( \lambda_{ij}  - \frac{b_{ij} - \sum_{k \ne j} a_{kj} \tilde{\lambda}_{ik} }{a_{jj}} \right)^2  + \frac{\psi_i\rho}{a_{jj}} P( |\lambda_{ij}|). \label{lambdaupdate} 
\end{eqnarray}
This is equivalent to minimizing the following penalized squared-error loss function 
\begin{equation*}
S(\tilde{\theta}) = {\rm arg} \min_{\theta} \left\{ \frac{1}{2}(\theta - \tilde{\theta})^2 + \rho^* P(|\theta|) \right\}. \label{lamdba_update_CD}
\end{equation*}
The solution $S(\tilde{\theta}) $ can be expressed in a closed form for a variety of convex and nonconvex penalties as follows:
\begin{description}
\item[lasso:] 
\begin{equation}
S(\tilde{\theta})={\rm sgn}(\tilde{\theta}) (|\tilde{\theta} |- \rho^*)_+. \label{uelasso}
\end{equation}
\item[SCAD:] 
\begin{equation*}
S(\tilde{\theta})=\left\{
\begin{array}{ll}
{\rm sgn}(\tilde{\theta})(|\tilde{\theta}| - \rho^*)_+ & \mbox{if $ |\tilde{\theta}| \le 2\rho^*$}\\
\dfrac{(\gamma-1)\tilde{\theta} - {\rm sgn}(\tilde{\theta})\rho^*\gamma}{\gamma-2} & \mbox{if $2\rho^* < |\tilde{\theta}| \le \rho^*\gamma$}\\
\tilde{\theta} & \mbox{if $|\tilde{\theta}| > \rho^*\gamma$}.\\
\end{array}
 \right.
\end{equation*}
\item[MC+:]
\begin{equation*}
S(\tilde{\theta})=\left\{
\begin{array}{ll}
 \dfrac{{\rm sgn}(\tilde{\theta})(|\tilde{\theta}| - \rho^*)_+}{1-1/\gamma} & \mbox{if $ |\tilde{\theta}| \le \rho^*\gamma$}\\
\tilde{\theta} & \mbox{if $|\tilde{\theta}| > \rho^*\gamma$}.\\
\end{array}
 \right.
\end{equation*}
\end{description}


After updating $\bm{\Lambda}$ by the coordinate descent algorithm, the new value of $\bm{\Psi}_{\rm new}$ is obtained by maximizing the expected penalized log-likelihood function in (\ref{ECL}) as follows:
\begin{equation*}
(\psi_i )_{\rm new} = s_{ii} - 2 (\hat{\bm{\lambda}}^T_i)_{\mathrm{new}}\bm{b}_i  +  (\hat{\bm{\lambda}}_i)_{\mathrm{new}}^T \bm{A} (\hat{\bm{\lambda}}_i)_{\mathrm{new}} \quad  \mbox{for $i=1,\dots,p$,}
\end{equation*}
where $(\psi_i )_{\rm new} $ is the $i$-th diagonal element of $\bm{\Psi}_{\rm new}$ and $(\hat{\bm{\lambda}}_i)_{\mathrm{new}}$ is the $i$-th column of $\hat{\bm{\Lambda}}_{\rm new}$.

We found that some of the column vectors of the factor loadings can be the zero vector when $\rho$ is sufficiently large.  As the value of $\rho$ decreases, the number of nonzero column vectors increases.  The following lemma describes the condition of each column of factor loadings.  
\begin{Lem}\label{lemmacolumn}
Each column of factor loadings computed by (\ref{problem_rotation_pmle3_adj}) cannot have only one nonzero parameter.
\end{Lem}
\begin{proof}
The proof is given in Appendix B.
\end{proof}

\subsection{Pathwise Algorithm}
We introduce a pathwise algorithm via coordinate descent that produces the entire solution path for the penalized likelihood factor analysis.  The pathwise algorithm can produce the solution for the grid of  increasing $\rho$ values $P=\{\rho_1,\dots,\rho_K  \}$ and $\gamma$  values $\Gamma=\{ \gamma_1,\dots,\gamma_T \}$.  Here $\rho_K$ is the smallest value for which the estimates of factor loadings satisfy $\hat{\bm{\Lambda}} \approx \bm{O}$, and $\gamma_T$ gives the lasso penalty (e.g., $\gamma_T=\infty$ for MC+ family).   

In the pathwise algorithm, first, the lasso solution path can be produced by decreasing the sequence of values for $\rho$, starting with $\rho=\rho_K$. Next, the value of $\gamma_{T-1}$ is selected and the solutions are produced for the sequence of $P=\{\rho_1,\dots,\rho_K  \}$.  
The  solution at ($\gamma_{T-1},\rho_k$) can be computed by using the solution at  ($\gamma_{T},\rho_k$), which leads to improved and smoother objective value surfaces \citep{Mazumderetal:2009}. In the same way, for $t=T-2,\dots,1$, the solution at ($\gamma_{t},\rho_k$) can be computed by using the solution at  ($\gamma_{t+1},\rho_k$).

\subsubsection{Initial Value of Factor Loading }\label{Initial update}
The  initial value of the factor loading and the unique variances might be  assumed to be $\bm{\Lambda} = \bm{O}$ and $\bm{\Psi} = {\rm diag} \bm{S}$.  With these initial values, however, the factor loadings cannot be updated with the EM algorithm no matter what value of $\rho$ is chosen, because $\bm{\Lambda}=\bm{O}$ is the stationary point of the log-likelihood function, i.e., 
\begin{equation*}
\left. \frac{\partial \ell(\bm{\Lambda},\bm{\Psi})}{\partial \bm{\Lambda}} \right|_{ \bm{\Lambda}=\bm{O}}  = \left.-N\bm{\Sigma}^{-1} (\bm{\Sigma} - \bm{S}) \bm{\Sigma}^{-1}\bm{\Lambda} \right|_{\bm{\Lambda}=\bm{O}} = \bm{O}
\end{equation*}
Therefore, the nonzero initial values, say $\breve{\bm{\Lambda}}=(\breve{\lambda}_{ij})$, must be determined.    It may be reasonable to assume that initial values for only the first column of factor loadings have nonzero elements, because most of the columns of factor loadings may be zero when $\rho$ is sufficiently large.  The initial values of the factor loading for the first column are defined by the maximum likelihood estimates of the one-factor model.   

\subsubsection{Selection of $\rho_K$}\label{selection of rhoK}
Next, the value of  $\rho_K$ is selected.  Since the factor loadings can be very sparse when  $\rho \approx \rho_K$,   the estimates of ${\bm{\Lambda}}$ at  $\rho \approx \rho_K$, say $\hat{\bm{\Lambda}}^{(h)}=(\hat{\lambda}_{ij}^{(h)})$,  may be close to
\begin{equation*}
\hat{\lambda}_{ij}^{(h)}=
\left\{
\begin{array}{ll}
\xi^{(h)}\breve{\lambda}_{\alpha 1} &i=\alpha , \ j=1\\
0&{\rm otherwise}
\end{array}
\right. \quad {\rm for} \  h=1,\dots,H.\label{initialLambda}
\end{equation*}
where $\alpha = {\rm arg } \max_i |\breve{\lambda}_{i1}|$, and $\xi^{(h)}$ ($h=1,\dots,H$)  is the scale parameter of the initial value. Typically, $H=10$ and $\xi^{(h)} = 0.1h$.  
As shown in Lemma \ref{lemmacolumn}, the first column of the factor loading should have at least two nonzero elements,  
so that 
 there exist nonzero elements of the factor loading $\hat{\lambda}_{i1}$ $(i \ne \alpha )$.  We can see from (\ref{uelasso}) that $\hat{\lambda}_{i1}$ $(i \ne \alpha )$ will stay zero if $|\tilde{\theta}| < \rho^*$, and thus define $\rho^{(h)}$  by 
$\rho^{(h)} = \max_{i \ne \alpha } {|b_{i1}|} / {\hat{\psi}_{i}},$ 
where $\hat{\psi}_{i}$ are the estimates of $\psi_{i}$ when the factor loadings are fixed by $\hat{\bm{\Lambda}}^{(h)}$.  Then, $\rho_K$ is defined  by $ \max_{h} \rho^{(h)}$.  

\subsubsection{Increasing  the Number of Factors}

The solution around $\rho_K$ may have nonzero elements for the first column but zeros for the other columns (i.e., the number of factors is 1).    When $\rho$ is small, the other columns should have nonzero parameters if $m \ge 2$, whereas the pathwise algorithm based on (\ref{lambdaupdate}) does not often produce nonzero elements for the other columns.  Therefore, at each $\rho$, we should check if the number of factors is increased. Let $m_0$ be the number of columns where some of the elements have nonzero estimates.   If $m_0<m$, the initial loading matrix is randomly generated and the model parameters are estimated by penalized likelihood procedure.  We adopt the newly estimated model if it yields larger value of penalized log-likelihood function.

\subsubsection{Reparameterization of the Penalty Function}\label{reparameterization:section}
The value of $\rho_K$, which is the smallest value that the factor loadings are set to approximately zero, is determined by Section \ref{selection of rhoK} for the lasso penalty. For the hard-thresholdings operater, however, the value of $\rho_K$ should be larger than that of the lasso. More generally, the value of $\rho_k$ should be monotonically increased as one moves across the family from the soft thresholdings operator to the hard one.  A reparameterization of the nonconvex penalty based on the degrees of freedom (\citealp{Ye:1998}; \citealp{Efron:1986}; \citealp{Efron:2004}) has been proposed by \citet{Mazumderetal:2009}, and  we apply it to the penalized likelihood factor analysis.  The reparameterization of the penalty function constrains the degrees of freedom at any value of $\rho$ to be constant as $\gamma$ changes.   

For example, the reparameterization of the MC+ family is as follows:  suppose that $\rho$ is the regularization parameter for $\gamma=\infty$, and $\rho^*=\rho^*(\rho, \gamma_0)$ is the regularization parameter for $\gamma=\gamma_0$.  The reparameterization can be realized by solving the following problem \citep{Mazumderetal:2009}:
\begin{eqnarray*}
\Phi(\gamma_0\rho) - \gamma_0\Phi(\rho) = -(\gamma_0-1)\Phi(\rho^*). \label{reparameterization}
\end{eqnarray*}


\subsection{Selection of the Regularization Parameter}
In this modeling procedure, it is important to select the appropriate value of the regularization parameter $\rho$.  
 The selection of the regularization parameter can be viewed as a model selection and evaluation problem. In regression analysis, the degrees of freedom of the lasso \citep{Zouetal:2007} may be used for selecting the regularization parameter.  With the use of the degrees of freedom, the following model selection criteria are introduced:
\begin{eqnarray*}
{\rm AIC} &=& -2 \ell(\hat{\bm{\Lambda}} ,\hat{\bm{\Psi}}) + 2\{df(\rho_k)+p\},\\
{\rm BIC} &=& -2 \ell(\hat{\bm{\Lambda}} ,\hat{\bm{\Psi}}) + (\log N ) \{df(\rho_k)+p\} ,\\
{\rm CAIC} &=& -2 \ell(\hat{\bm{\Lambda}} ,\hat{\bm{\Psi}}) +  (\log N + 1) \{df(\rho_k)+p\},
\end{eqnarray*}
where $df(\rho_k)$ is the number of nonzero parameters for the lasso penalty at $\rho=\rho_k$. Note that this formula can be applied to any value of $\gamma$, because the reparameterization of the penalty function described in Section \ref{reparameterization:section} constrains the degrees of freedom to be constant as $\gamma$ varies.  

  \subsection{Treatment for Improper Solutions}
 It is well-known that the maximum likelihood estimates of unique variances can turn out to be zero or negative, which is referred as the improper solutions, and many researchers have studied this problem (e.g., \citealp{van1978various,anderson1984effect,kano1998improper}).  In general, the occurrence of improper solutions makes converge of the algorithm slow and unstable.  In order to handle this issue, we add a penalty with respect to $\mathbf{\Psi}$ to (\ref{problem_rotation_pmle3_adj}) according to the basic idea given by \citet{martin1975bayesian} and \citet{hirose2011bayesian}: 
\begin{equation*}
\ell_{\rho}^*(\mathbf{\Lambda},\mathbf{\Psi})  =  \ell_{\rho}(\mathbf{\Lambda},\mathbf{\Psi})   -\frac{N}{2} \eta{\rm tr}(\mathbf{\Psi}^{-1/2}\mathbf{S}\mathbf{\Psi}^{-1/2}),  \label{problem_rotation_pmle3_adj2}
\end{equation*}
where $\eta$ is a tuning parameter.  Note that when  $\psi_i \rightarrow 0$, ${\rm tr}(\mathbf{\Psi}^{-1/2}\mathbf{S}\mathbf{\Psi}^{-1/2}) \rightarrow \infty$.  Thus, the penalty term ${\rm tr}(\mathbf{\Psi}^{-1/2}\mathbf{S}\mathbf{\Psi}^{-1/2})$ prevents the occurrence of improper solutions.  \citet{hirose2011bayesian} derived a generalized Bayesian information criterion \citep{konishi2004bayesian} for selecting the appropriate value of $\eta$, whereas it is difficult to derive generalized Bayesian model criterion in lasso-type penalization procedure.  In practice, the penalty term can prevent the occurrence of improper solution even when $\eta$ is very small such as 0.001.  

\section{Numerical Examples}
\subsection{Monte Carlo Simulations}\label{sec:simulation}
In the simulation study, we used two models according to the following factor loadings:
\begin{eqnarray*}
{\bf Model (A): \quad} 
\bm{\Lambda} &=& 
\left(
\begin{array}{rrrrrr}
0.95&0.90&0.85&0.00&0.00&0.00\\
0.00&0.00&0.00&0.80&0.75&0.70\\
\end{array}
\right)^T,\\
{\bf Model (B): \quad } 
 \mathbf{\Lambda} &=& 
 \left(
 \begin{array}{cccc}
 0.95 \cdot \mathbf{1}_{250} &\mathbf{0}_{250} & \mathbf{0}_{250} &\mathbf{0}_{250} \\
\mathbf{0}_{250} &0.90\cdot\mathbf{1}_{250} & \mathbf{0}_{250}  & \mathbf{0}_{250}  \\
\mathbf{0}_{250} &\mathbf{0}_{250} &0.85 \cdot\mathbf{1}_{250}  & \mathbf{0}_{250}  \\
\mathbf{0}_{250} &\mathbf{0}_{250} & \mathbf{0}_{250}  & 0.80 \cdot\mathbf{1}_{250}  \\
 \end{array}
 \right),
 \end{eqnarray*}
where $\mathbf{1}_{q}$ is a $q$-dimensional vector with each element being 1, and $\mathbf{0}_{q}$ is a $q$-dimensional zero vector.  For both Models (A) and (B), we set $\mathbf{\Psi} = {\rm diag}( \mathbf{I}_p - \mathbf{\Lambda}\mathbf{\Lambda}^T)$. 

For each model, 1000 data sets were generated with $\bm{x} \sim N(\bm{0},\bm{\Lambda}\bm{\Lambda}^T + \bm{\Psi})$.     The number of observations was $N=50, 100$, and $200$.  The model was estimated by the penalized maximum likelihood method  via the MC+ family with $\gamma=1.96$, and the lasso.  The regularization parameter was selected by the AIC, BIC and CAIC.  For Model (A), we also applied the traditional two-step estimation procedure, i.e., the model was estimated by the maximum likelihood method and then the following rotation techniques were employed: lasso-type loss function described in Example \ref{example:pss}, and the varimax rotation method \citep{kaiser1958varimax}.  The two-step traditional procedure was not applied to Model (B), because the maximum likelihood estimates were not available for $N<<p$ case.   

Tables \ref{table:simulation1} and \ref{table:simulation2} show the mean squared error (MSE) of $\bm{\Lambda}$ and $\bm{\Psi}$, and the true positive rate (TPR) and true negative rate (TNR)  over 1000 simulations.   The MSE of $\bm{\Lambda}$ and $\bm{\Psi}$ are defined by
\begin{eqnarray*}
{\rm MSE}_{\bm{\Lambda}} =\frac{1}{1000pm} \sum_{s=1}^{1000} \| \bm{\Lambda} - \hat{\bm{\Lambda}}^{(s)} \|^2 \mbox{ and } {\rm MSE}_{\bm{\Psi}} = \frac{1}{1000p} \sum_{s=1}^{1000} \| \bm{\Psi} - \hat{\bm{\Psi}}^{(s)} \|^2, 
\end{eqnarray*}
where $\hat{ \bm{\Lambda}}^{(s)} $ and $\hat{ \bm{\Psi}}^{(s)} $ are the estimates for the $s$-th dataset.  TPR (TNR)  indicates the proportion of cases where non-zero (zero) factor loadings correctly set to non-zero (zero).  
From the results of Tables \ref{table:simulation1} and \ref{table:simulation2}, we obtain the following empirical observations:
\begin{itemize}
\item When the number of observations $N$ increased, the MSE became small and the values of TNR and TPR became large.
\item All methods often detected the non-zero elements correctly (i.e.,  TPR was close to 1).
\item The MC+ family often performed much better than the lasso in terms of TNR:  the MC+ was able to produce sparser models than the lasso and also detected the zero elements correctly.   
\item For Model (A), the lasso penalization method was able to produce sparser solutions than the lasso rotation.    
\item For Model (B), the MC+ family yielded much smaller MSE than the lasso.    
\item   The AIC often selected too dense model, because the TNR was small compared with the BIC and CAIC.
\end{itemize}

\begin{table}[!t]
\caption{Mean squared error of the factor loadings and uniquenesses, and the true positive rate (TPR), true negative rate (TNR) for Model (A).  The terms $\rm rot_{lasso}$ and $ \rm rot_{varimax} $ in the last two columns represent the rotation method via the lasso penalty and the varimax rotation technique, respectively. } \label{table:simulation1}
\begin{center}
\begin{tabular}{lrrrrrrrrrr}
  \hline
 & \multicolumn{6}{c}{Penalization Methods}&  \multicolumn{2}{c}{Rotation Methods} \\ 
 & \multicolumn{2}{c}{AIC}&  \multicolumn{2}{c}{BIC}&  \multicolumn{2}{c}{CAIC} &---&---  \\ 
  & MC+ & lasso   & MC+ & lasso   & MC+ & lasso & $\rm rot_{lasso}$ & $ \rm rot_{varimax} $\\ 
  \hline
$N=50$ & \\
${\rm MSE}_{\bm{\Lambda}} \times 10^1$ & 1.04 & 1.13 & 1.65 & 1.46 & 2.89 & 1.87 & 1.07 & 0.98 \\ 
${\rm MSE}_{\bm{\Psi}} \times 10^1$&0.92 & 0.86 & 1.36 & 0.92 & 2.08 & 1.02 & 0.84 & 0.84 \\
  TPR& 1.00 & 1.00 & 0.98 & 1.00 & 0.94 & 0.99 & 1.00 & 1.00 \\
  TNR &0.70 & 0.41 & 0.80 & 0.50 & 0.85 & 0.55 & 0.15 & 0.00 \\ 
   \hline
   $N=100$ & \\
${\rm MSE}_{\bm{\Lambda}} \times 10^1$ &  0.42 & 0.54 & 0.40 & 0.72 & 0.45 & 0.84 & 0.50 & 0.46 \\ 
${\rm MSE}_{\bm{\Psi}} \times 10^1$&0.42 & 0.40 & 0.48 & 0.41 & 0.53 & 0.42 & 0.39 & 0.39 \\ 
  TPR& 1.00 & 1.00 & 1.00 & 1.00 & 1.00 & 1.00 & 1.00 & 1.00 \\ 
  TNR &   0.79 & 0.41 & 0.89 & 0.54 & 0.91 & 0.58 & 0.16 & 0.00 \\ 
   \hline
$N=200$ & \\
${\rm MSE}_{\bm{\Lambda}} \times 10^1$ &0.17 & 0.26 & 0.12 & 0.39 & 0.13 & 0.46 & 0.24 & 0.22 \\ 
${\rm MSE}_{\bm{\Psi}} \times 10^1$&  0.19 & 0.19 & 0.20 & 0.20 & 0.21 & 0.20 & 0.19 & 0.19 \\ 
  TPR&  1.00 & 1.00 & 1.00 & 1.00 & 1.00 & 1.00 & 1.00 & 1.00 \\ 
  TNR & 0.87 & 0.42 & 0.96 & 0.57 & 0.97 & 0.61 & 0.15 & 0.00 \\ 
   \hline
\end{tabular}
\end{center}
\end{table}

\begin{table}[!t]
\caption{Mean squared error of the factor loadings and uniquenesses, and the true positive rate (TPR), true negative rate (TNR) for Model (B).  } \label{table:simulation2}
\begin{center}
\begin{tabular}{lrrrrrrrr}
  \hline
 & \multicolumn{2}{c}{AIC}&  \multicolumn{2}{c}{BIC}&  \multicolumn{2}{c}{CAIC}\\
   & MC+ & lasso   & MC+ & lasso   & MC+ & lasso & \\
     \hline
$N=50$ & \\
${\rm MSE}_{\bm{\Lambda}} \times 10^{-1} $ &   2.29 & 4.06 & 2.74 & 9.29 & 3.36 & 15.7 \\
${\rm MSE}_{\bm{\Psi}} $& 3.81 & 3.65 & 4.09 & 3.65 & 5.07 & 3.65 \\
  TPR&  1.00 & 1.00 & 0.96 & 1.00 & 0.92 & 1.00  \\
  TNR &0.44 & 0.01 & 0.70 & 0.02 & 0.95 & 0.03 \\
     \hline
   $N=100$ & \\
${\rm MSE}_{\bm{\Lambda}} \times 10^{-1} $    & 0.91 & 2.18 & 0.78 & 7.21 & 0.76 &11.1 \\
${\rm MSE}_{\bm{\Psi}}$&1.87 & 1.80 & 1.93 & 1.80 & 1.93 & 1.80\\
  TPR& 1.00 & 1.00 & 1.00 & 1.00 & 1.00 & 1.00 \\
  TNR &  0.57 & 0.01 & 0.95 & 0.02 & 0.98 & 0.03\\
   \hline
$N=200$ & \\
${\rm MSE}_{\bm{\Lambda}} \times 10^{-1} $    & 0.45 & 1.28 & 0.32 & 4.74 & 0.32 & 6.94  \\
${\rm MSE}_{\bm{\Psi}}$& 0.93 & 0.90 & 0.93 & 0.90 & 0.93 & 0.90\\
  TPR& 1.00 & 1.00 & 1.00 & 1.00 & 1.00 & 1.00 \\
  TNR &   0.84 & 0.01 & 1.00 & 0.02 & 1.00 & 0.03\\
   \hline
\end{tabular}
\end{center}
\end{table}

\subsection{Analysis of Handwritten Data} \label{sec:Handwritten}
We illustrate the usefulness of the proposed procedure through a de-noising experiment of handwritten data available at \url{http://www-stat.stanford.edu/~tibs/ElemStatLearn/data.html}.   The dataset consists of 652 observations of the digit of ``4" made from 16 $\times$ 16 pixels scaled between $-1$ and $1$.  In order to evaluate the robustness of our proposed procedure, we randomly added the noise from $U(-1,1)$ to the pixels  randomly chosen with probability 0.1.  The dimensionality of the data was reduced to $m = 10$ so that the percentage of non-zero loadings was approximately 20\%, 40\%, 60\%, and  80\%.  We compared three data compression and reconstruction  techniques: penalized likelihood factor analysis via the lasso, penalized likelihood factor analysis via the MC+ with $\gamma=1.96$, and the sparse principal component analysis (Sparse PCA, \citealp{zou2006sparse}).

For Sparse PCA, the data reconstruction of data $\bm{x}_i$ was made by $\bm{\Lambda}(\bm{\Lambda}^T\bm{\Lambda})^{-1}\bm{\Lambda}^T\bm{x}_i$, where $\bm{\Lambda}$ is the sparse loading matrix.  
For factor analysis,  the data was reconstructed via the posterior mean: $\bm{\Lambda}E[\bm{F}_i | \bm{x}_i] = \bm{\Lambda}\bm{M}^{-1}\bm{\Lambda}^T\bm{\Psi}^{-1}\bm{x}_i $.   The performance of the three methods to data compressions and reconstruction was evaluated by the MSE between the true data without noise and reconstructed data. Figure \ref{fig:handwritten} shows the reconstruction error (left panel), some of the digit images of original test data with noise, and reconstructed images with percentage of non-zero loadings being 40\% (right panel).  From the right panel of Figure \ref{fig:handwritten}, the penalized likelihood factor analysis produced smoother images than the Sparse PCA.  Furthermore, the MC+ yielded the smallest MSE among the methods when the percentage of non-zero loadings was small.   When the estimated loading matrix became dense, the three methods yielded almost same values of MSE.  

\begin{figure}[t]
\begin{center}
    \includegraphics[width=160mm]{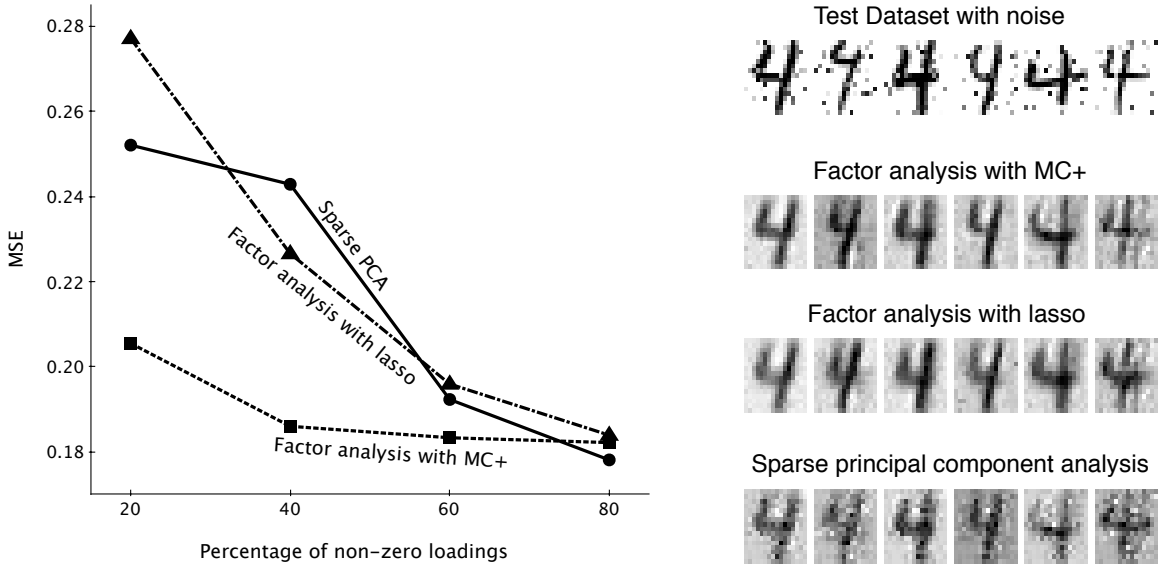}
\caption{The reconstruction error (left panel), some of the digit images of original test data with noise, and reconstructed images with percentage of non-zero loadings being 40\% (right panel).}
\label{fig:handwritten}
\end{center}
\end{figure}

\section{Concluding Remarks}
We have proposed a penalized maximum likelihood factor analysis via nonconvex penalties, and shown that  the proposed methodology can produce sparser solutions than the traditional rotation techniques.  A new algorithm via the EM algorithm and coordinate descent was presented, which produces the entire solution path for a wide variety of convex and nonconvex penalties.  Monte Carlo simulations were conducted to investigate the effectiveness of the proposed procedure.  Although the lasso can yield sparse solutions in factor analysis models, the MC+ often produced sparser solutions than the lasso, so that the true model structure can often be reconstructed.   The proposed procedure was applied to handwritten data, which showed that the MC+ was able to perform better than the lasso and Sparse PCA when the estimated loading matrix was sufficiently sparse.    

The proposed algorithm can be applied to high-dimensional data such as $p=10000$, whereas sometimes it takes several hours to compute the entire solution path.  As a future research topic, it would be interesting to propose a much faster algorithm than the EM with coordinate descent. 
In the present paper, the tuning parameter $\rho$ was selected by the information criteria based on the degrees of freedom of the lasso.  However, the degrees of freedom of the lasso are defined in the framework of the regression model.  Another interesting topic is to provide a mathematical support for the degrees of freedom of the lasso in  factor analysis model.  

\appendix
\def\thesection{Appendix \Alph{section}:}
\renewcommand{\theequation}{A.\arabic{equation}}
\section{Derivation of Complete-Data Penalized Log-Likelihood Function in EM Algorithm}
 In order to apply the EM algorithm, first, the common factors $\bm{f}_n$ can be regarded as missing data and maximize the complete-data penalized log-likelihood function
\begin{equation*}
\l_{\rho}^{C} (\bm{\Lambda},\bm{\Psi}) = \sum_{n=1}^N \log f(\bm{x}_n,\bm{f}_n) - N   \sum_{i=1}^p\sum_{j=1}^m \rho P(|\lambda_{ij}|), \label{g_C}
\end{equation*}
where the density function $f(\bm{x}_n,\bm{f}_n)$ is defined by
\begin{equation*}
f(\bm{x}_n,\bm{f}_n) = \prod_{i=1}^p \left\{  (2\pi\psi_i)^{-1/2} \exp \left( - \frac{ (x_{ni}-\bm{\lambda}_i^T\bm{f}_n )^2}{2\psi_i}  \right) \right\}  (2\pi)^{-m/2}\exp \left( - \frac{\| \bm{f}_n \|^2}{2} \right)
\end{equation*}
Then, the expectation of  $\l_{\rho}^{C} $ can be taken with respect to the distributions $f(\bm{f}_n | \bm{x}_n,\bm{\Lambda},\bm{\Psi})$,
\begin{eqnarray*}
E[\l_{\rho}^{C} (\bm{\Lambda},\bm{\Psi})] &=&-\frac{N(p+m)}{2} \log(2\pi) - \frac{N}{2} \sum_{i=1}^p \log \psi_i \\
&&- \frac{1}{2} \sum_{n=1}^N\sum_{i=1}^p \frac{x_{ni}^2 - 2x_{ni}\bm{\lambda}_i^TE[\bm{F}_n|\bm{x}_n]+ \bm{\lambda}_i^T E[\bm{F}_n\bm{F}_n^T|\bm{x}_n]\bm{\lambda}_i}{\psi_i}  \\
&&- \frac{1}{2} {\rm  tr} \left \{ \sum_{n=1}^N E[\bm{F}_n\bm{F}_n^T|\bm{x}_n] \right\} - N   \sum_{i=1}^p\sum_{j=1}^m \rho P(|\lambda_{ij}|)
\end{eqnarray*}
For given $\bm{\Lambda}_{\rm old}$ and $\bm{\Psi}_{\rm old}$, the posterior  $f(\bm{f}_n | \bm{x}_n,\bm{\Lambda}_{\rm old}, \bm{\Psi}_{\rm old})$ is normally distributed with $E[\bm{F}_n|\bm{x}_n] = \bm{M}^{-1}\bm{\Lambda}_{\rm old}^T\bm{\Psi}_{\rm old}^{-1} \bm{x}_n$ and $E[\bm{F}_n\bm{F}_n^T|\bm{x}_n] = \bm{M} ^{-1} + E[\bm{F}_n|\bm{x}_n] E[\bm{F}_n|\bm{x}_n] ^T$, where $\bm{M} = \bm{\Lambda}_{\rm old}^T\bm{\Psi}_{\rm old}^{-1}\bm{\Lambda}_{\rm old} + \bm{I}_m$.  Then, we have
\begin{eqnarray*}
\sum_{n=1}^N E[\bm{F}_n]x_{ni} &=& \sum_{n=1}^N \bm{M}^{-1}\bm{\Lambda}_{\rm old}^T\bm{\Psi}_{\rm old}^{-1}\bm{x}_nx_{ni}=N\bm{M}^{-1}\bm{\Lambda}_{\rm old}^T\bm{\Psi}_{\rm old}^{-1}\bm{s}_i,\\
\sum_{n=1}^N E[\bm{F}_n\bm{F}_n^T] &=& \sum_{n=1}^N (\bm{M} ^{-1} + \bm{M}^{-1}\bm{\Lambda}_{\rm old}^T\bm{\Psi}_{\rm old}^{-1}\bm{x}_n\bm{x}_n^T\bm{\Psi}_{\rm old}^{-1}\bm{\Lambda}_{\rm old}\bm{M}^{-1})\\
&=&N (\bm{M} ^{-1} + \bm{M}^{-1}\bm{\Lambda}_{\rm old}^T\bm{\Psi}_{\rm old}^{-1}\bm{S}\bm{\Psi}_{\rm old}^{-1}\bm{\Lambda}_{\rm old}\bm{M}^{-1}),
\end{eqnarray*}
Let $\bm{M}^{-1}\bm{\Lambda}_{\rm old}^T\bm{\Psi}_{\rm old}^{-1}\bm{s}_i$ and $\bm{M} ^{-1} + \bm{M}^{-1}\bm{\Lambda}_{\rm old}^T\bm{\Psi}_{\rm old}^{-1}\bm{S}\bm{\Psi}_{\rm old}^{-1}\bm{\Lambda}_{\rm old}\bm{M}^{-1}$ be  $\bm{b}_i$ and $\bm{A}$, respectively.  Then, the expectation of  $\l_{\rho}^{C} $ in (8) can be derived.

\section{Proof of Lemma 4.1}
The proof is by contradiction.  Assume that $\hat{\bm{\Lambda}}$ and $\hat{\bm{\Psi}}$ are the solution of (7) and $j$-th column of $\hat{\bm{\Lambda}}$ has only one nonzero element, say, $\hat{\lambda}_{aj}$. Another parameter $\hat{\bm{\Lambda}}^*$ and $\hat{\bm{\Psi}}^*$ are defined, where $\hat{\bm{\Lambda}}^*$  is same as $\hat{\bm{\Lambda}}$ but with $(a,j)$-th element being zero and $\hat{\bm{\Psi}}^*$  is same as $\hat{\bm{\Psi}}$  but with $j$-th diagonal element being $\hat{\psi}_j+\hat{\lambda}_{aj}^2$.  In this case,  we have the same covariance structure, i.e., $\hat{\bm{\Lambda}}\hat{\bm{\Lambda}}^T+\hat{\bm{\Psi}}=\hat{\bm{\Lambda}}^*\hat{\bm{\Lambda}}^{*T}+\hat{\bm{\Psi}}^*$, which suggests $\ell(\hat{\bm{\Lambda}}, \hat{\bm{\Psi}}) = \ell(\hat{\bm{\Lambda}}^*, \hat{\bm{\Psi}}^*)$, whereas the penalty term of $ \sum_{i=1}^p\sum_{j=1}^m \rho P(|\hat{\lambda}_{ij}|)$ is larger than $\sum_{i=1}^p\sum_{j=1}^m \rho P(|\hat{\lambda}^*_{ij}|)$.  This means $\ell_{\rho}(\hat{\bm{\Lambda}}, \hat{\bm{\Psi}}) < \ell_{\rho}(\hat{\bm{\Lambda}}^*, \hat{\bm{\Psi}}^*)$, which contradicts the assumption that $\hat{\bm{\Lambda}}$ and $\hat{\bm{\Psi}}$ are penalized maximum likelihood estimates.  

\bibliographystyle{ECA_jasa}
\bibliography{paper-ref}

\end{document}